\documentclass[journal]{IEEEtran}
\ifCLASSINFOpdf
\else
\fi
\usepackage{mathrsfs}
\usepackage{amsfonts}
\usepackage{amsmath,amssymb}
\usepackage{graphicx}
\usepackage{float}
\usepackage[dvips]{epsfig}
\usepackage[dvips]{color}
\usepackage{subfigure}
\usepackage{caption}
\usepackage{graphics,graphicx}
\usepackage{pstricks,pst-node,pst-tree,pstricks-add}

\def\dref#1{(\ref{#1})}

\newtheorem{assumption}{Assumption}
\newtheorem{lemma}{Lemma}
\newtheorem{theorem}{Theorem}

\newtheorem{remark}{Remark}

\newtheorem{proof}{Proof}

\begin{document}
%
\title{Event-Triggered Consensus of Homogeneous and Heterogeneous Multi-Agent Systems with Jointly Connected Switching Topologies}
%

\author{Bin~Cheng, Xiangke Wang,~\IEEEmembership{Senior Member,~IEEE},
        and~Zhongkui~Li,~\IEEEmembership{Member,~IEEE}
\thanks{This work was supported in part by the National Natural Science Foundation of China under grants 61473005, 61403406, U1713223, 11332001, and by Beijing Nova Program under grant 2018047.}
\thanks{B. Cheng and Z. Li are with the State Key Laboratory for Turbulence and Complex Systems, Department of Mechanics and Engineering Science, College of Engineering, Peking University, Beijing 100871, China.
E-mail: \tt \{bincheng,zhongkli\}@pku.edu.cn}
\thanks{X. Wang is with the College of Mechatronic Engineering and Automation, National University of Defense Technology, Changsha, China.
E-mail: {\tt xkwang@nudt.edu.cn}}
\thanks{Corresponding author: Zhongkui Li.}}

\maketitle

\begin{abstract}
This paper investigates the distributed event-based consensus problem of switching networks satisfying the jointly connected condition.
Both the state consensus of homogeneous linear networks and output consensus of heterogeneous networks are studied.
Two kinds of event-based protocols based on local sampled information are designed, without the need to solve any matrix equation or inequality.
Theoretical analysis indicates that the proposed event-based protocols guarantee the achievement of consensus and the exclusion of Zeno behaviors for jointly connected undirected switching graphs.
These protocols, relying on no global knowledge of the network topology and independent of switching rules, can be devised and utilized in a completely distributed manner.
They are able to avoid continuous information exchanges for either controllers' updating or triggering functions' monitoring, which ensures the feasibility of the presented protocols.
\end{abstract}


%
\IEEEpeerreviewmaketitle

\section{introduction}
Event-driven coordination has been widely studied and started maturing to soon stand alone
in the control area
in the last decade \cite{Lewis2014cooperative,z-k-li2010consensus,
ZLi2015designing,ZLi2014cooperative,zhongkui17robust,b-a-khashooei2017output,
y-q-wu2016an,xiao-hua-ge2017distributed,lei-ding2018distributed}.
Compared to classic continuous control approaches, event-based control has numerous advantages especially in enhancing control efficiency, such as avoiding continuously updating controllers and continuous communications among neighboring agents.
The latter advantage is particularly evident when we focus on Internet of Things and other large-scale networks where the cyber operations, including processing, storage, and communication, must be viewed as a scare, globally shared resource \cite{Nowzari17Event}.
Due to these practical considerations, it is not surprising that so many researchers are interested in event-triggered control and present plenty of results.
Applying event-driven control in networked systems poses some new challenges that do not exist in either area alone \cite{Nowzari17Event}.
As pointed out in \cite{Nowzari17Event}, researchers must consider how to deal with the natural asynchronism introduced into the systems and how to rule out the Zeno behavior.
Another challenge is that the separation principle cannot be used for event-triggered control systems anymore \cite{CRamesh2011}.

Existing works have presented a large number of insights into general coordination of networked systems with event-triggered mechanisms.
As a specific case study, event-triggered consensus is a longstanding area of research in multi-agent systems; see the references
\cite{Henningsson2008sporadic,Astrom1999comparison,PTabuada2007event,
Heemels2008analysis,Heemels2012an,RZheng2016stability,BCheng2018fully,BCheng2018consensus}.
Many survey papers about event-driven control were published, such as \cite{Nowzari17Event,lei-ding2018an,xian-ming-zhang2017an,xian-ming-zhang2016survey}.
Generally speaking,
existing consensus protocols are designed for either state consensus of homogeneous networks or output consensus of heterogeneous networks.
Noting that for heterogeneous networks, where even the dimensions of states may be different, output consensus is a more meaningful topic than state consensus.

In the field of state consensus of homogeneous networks, \cite{Garcia2013decentralised,Meng2013event,DVDimarogonas2012distributed} presented event-based protocols for single-integrator agents under undirected graphs.
To remove the limitation that continuous information was still required in triggering functions of early works,
\cite{Seyboth2013event} proposed triggering functions only based on discrete sampled information.
The authors of \cite{DYang2016Decentralized,Zhu2014event,Guo2014a} presented event-driven consensus algorithms for general linear networks.
Reference \cite{HZhang2014observer} studied the event-driven consensus using output feedback control.
The event-based consensus control problem with external disturbances was studied in \cite{LXing2017event,JLiu2017fixed,xiao-hua-ge2017event}.
Event-driven output consensus of heterogeneous networks was studied in \cite{WHu2017output,xdliu2017distributed}.
The authors of \cite{WHu2017cooperative} studied event-based cooperative output regulation problem of heterogeneous networks.

It should be noted that the proposed protocols in the above works were only designed for fixed and connected topologies.
However, in many practical cases, the topologies may be switching \cite{WNi2010leader,Su2012output,xiao-hua-ge2017consensus,bo-da-ning2018collective} and do not satisfy the connected condition.
In \cite{adaldo2015event}, the authors proposed an event-driven protocol for networks with switching communication graphs.
One limitation of the protocol in \cite{adaldo2015event}, that the triggering functions were designed based on continuous information, may limit its practical applicability.
To avoid continuous interagent communication, \cite{thcheng2017event} proposed decentralized event-based controllers for leader-follower networks under fixed or switching graphs.
The results of \cite{thcheng2017event} relied on an assumption that the (switching) topology is connected at every moment, which was not always satisfied for general switching topologies.
In particular, there were even no any connections among agents at some special instants.
This assumption was removed by the authors of \cite{z-g-wu-event,wyxu2017event}, in which similar problems were considered.
The designs of the protocols proposed in \cite{z-g-wu-event,wyxu2017event}, nevertheless, required to solve two coupled inequalities, while the existence of the solution is unclear in general cases.
The switching nature of topologies coupled with event-triggered communications makes it troublesome to propose distributed consensus algorithms,
and the existence of heterogeneity renders the task for heterogeneous networks more challenging.
How to devise event-triggered consensus algorithms for linear homogeneous (or heterogeneous) networks with general switching topologies needs further investigation.

In the current paper, we study the event-driven consensus control problems with switching graphs, including state consensus of homogeneous linear networks and output consensus of heterogeneous linear networks.
For the homogeneous case,
we present an event-based protocol, composed of controllers and triggering rules.
Under this protocol, communications will not take place until the topology switches or the designed measurement error exceeds an appropriate threshold.
It is shown that state consensus is achieved and Zeno behaviors are ruled out.
The protocol can be explicitly constructed and do not need to solve any matrix equation or inequality.
We also consider event-based output consensus of heterogeneous networks with switching topologies and an exogenous signal that can be viewed as a reference input or an external disturbance.
For this problem, we first devise distributed observers to estimate the exogenous signal and then propose local control inputs.

The main contributions of this paper are listed as follows.
We have solved both the event-based state consensus control problem of homogeneous networks and the event-based output consensus control problem of heterogeneous networks.
Different from existing related papers, the proposed event-triggered protocols of this paper can be used for any switching graphs satisfying the jointly connected condition, including fixed graphs as a special case.
The proposed protocols, requiring no global information associated with the whole network and independent of the switching rules, can be devised and utilized in a completely distributed manner.
The Zeno behavior can be excluded at any finite time by showing that the interval between any different triggering instants is not less than a strictly positive value.
This feature ensures the feasibility of the above protocols when they are implemented on practical systems.

Here is the outline of this paper.
In Section \ref{section_state}, we consider the event-driven state consensus of homogeneous networks.
We then study event-based output consensus of heterogeneous networks in Section \ref{section_output}.
Numerical simulations and conclusions are presented in Sections \ref{s_sim} and \ref{s_con}, respectively.

%

\section{event-based state consensus of homogeneous multi-agent systems}\label{section_state}

\subsection{Problem Formulation}
In this section, we consider $N$ homogeneous linear agents, whose dynamics satisfy
\begin{equation}\label{sys_homo}
\begin {aligned}
\dot x_i=Ax_i+Bu_i,~i=1,\cdots,N,
\end{aligned}
\end{equation}
where $x_i\in \mathbf{R}^n$ denotes the state, $u_i\in \mathbf{R}^p$ represents the control input, and $A\in \mathbf{R}^{n\times n}$, $B\in\mathbf{R}^{n\times p}$ are constant matrices.

\begin{assumption}\label{assumption_A}
The pair $(A,B)$ is stabilizable and $A$ is neutrally stable \footnote{A matrix $A\in \mathbf C^{n\times n}$ is neutrally stable in the continuous-time sense if it has no eigenvalue with positive real part and the Jordan block corresponding to any eigenvalue on the imaginary axis is of size one, while is Hurwitz if all of its eigenvalues have strictly negative real parts \cite{ZLi2014cooperative}.}.
\end{assumption}

Denote $\theta:~[0,+\infty)\rightarrow \Theta$ as a switching signal with a positive dwelling time $\tau$.
Let $\mathcal {G}_{\theta(t)}\triangleq(\mathcal {V}, \mathcal
{E}_{\theta(t)})$ represent an undirected graph among the $N$ agents, where $\mathcal {V}=\{1,\cdots,N\}$ and $\mathcal {E}_{\theta(t)}\subseteq\mathcal {V}\times\mathcal
{V}$ denote the sets of nodes and edges, respectively.
Consider an infinite time sequence composed of nonempty, bounded, and contiguous intervals
$[\bar t_0,\bar t_1)$, $\cdots$, $[\bar t_k,\bar t_{k+1})$, $\cdots$, with $\bar t_0=0$.
Suppose $\bar t_{k+1}-\bar t_k\leq T$ with $T$ being some positive constant and during each interval $[\bar t_k,\bar t_{k+1})$, there are finite nonoverlapping subintervals
$$[\bar t_k^0,\bar t_k^1),[\bar t_k^1,\bar t_k^2),\cdots,[\bar t_k^{m_k-1},\bar t_k^{m_k}),~\bar t_k=\bar t_k^0,~\bar t_{k+1}=\bar t_k^m,$$
satisfying $\bar t_k^{j+1}-\bar t_k^j\geq \tau$, $j=0,1,\cdots,m_k-1$.
And $\mathcal G_{\theta(t)}$ is fixed during each subinterval.
An edge of $\mathcal E_{\theta(t)}$ is composed of two distinct nodes of $\mathcal V$. If $(i,j)\in\mathcal {E}_{\theta(t)}$, $i$
and $j$ are neighbors under graph $\mathcal G_{\theta(t)}$.
An undirected path between nodes $i$ and $j$ is denoted as $({i_1}, {i_{2}})$, $({i_2}, {i_{3}})$, $\cdots$, $({i_q}, j)$.
Denote the adjacency matrix of graph $\mathcal G_{\theta(t)}$ by $\mathcal A(t)=[a_{ij}(t)]\in \mathbf R^{N\times N}$, where $a_{ii}(t)=0$, $a_{ij}(t)=1$ if $(j,i)\in \mathcal E_{\theta(t)}$ and $a_{ij}(t)=0$ otherwise.
Denote the Laplacian matrix $\mathcal L_{\theta(t)}=[l_{ij}(t)]\in \mathbf R^{N\times N}$ of $\mathcal G_{\theta}$ by $l_{ii}(t)=\sum_{j=1}^N{a_{ij}(t)}$ and $l_{ij}(t)=-a_{ij}(t)$, $i\neq j$.
Define the degree as $d_i(t) = l_{ii}(t)$, $i\in\mathcal V$.
Then, define $\bigcup_{t\in[\bar t_k,\bar t_{k+1})}\mathcal G_{\theta(t)}$ as a union graph in the collection for time $t$ from $\bar t_k$ to $\bar t_{k+1}$.


\begin{assumption}\label{assumption_gr}
The undirected graph $\mathcal{G}_{\theta(t)}$ of the $N$ agents is jointly connected, i.e., $\bigcup_{t\in[\bar t_k,\bar t_{k+1})}\mathcal G_{\theta(t)}$ is connected.
\end{assumption}

The objective here is to present distributed event-based algorithms under which all subsystems described by \dref{sys_homo} converge to a common state trajectory
and Zeno behaviors can be eliminated.

Instead of using agents' actual states,
define the state estimate as $\tilde x_i(t)\triangleq e^{A(t-t_k^i)}x_i(t_k^i)$, $\forall t\in [t_k^i,t_{k+1}^i)$,~$i=1,\cdots,N$,
where $t_k^i$ denotes the $k$-th event instant of agent $i$.
The event instants $t_0^i$, $t_1^i$, $\cdots$, are determined by the triggering function to be designed later.
Using the relative state estimates of neighboring agents, we present a distributed event-based controller as:
\begin{equation}\label{pro1}
\begin {aligned}
u_i(t)=cG\sum_{j=1}^N{a_{ij}(t)(\tilde x_i-\tilde x_j)},~i\in\mathcal V,
\end{aligned}
\end{equation}
where $c>0$ and $G\in \mathbf R^{p\times n}$ are design parameters.

Define
$\xi=[\xi_1^T,\cdots,\xi_N^T]^T$
and $\tilde\xi=[\tilde\xi_1^T,\cdots,\tilde\xi_N^T]^T$ with
$\xi_i\triangleq x_i-\frac{1}{N}\sum_{j=1}^N{x_j}$ and
$\tilde\xi_i\triangleq \tilde x_i-\frac{1}{N}\sum_{j=1}^N{\tilde x_j}$, $i=1,\cdots,N$.
Letting $x\triangleq [x_1^T,\cdots,x_N^T]^T$ gives
$\xi=(M\otimes I_n)x$ and $\tilde \xi=(M\otimes I_n)\tilde x$,
where $M=I_N-\frac{1}{N}\mathbf 1_N\mathbf 1_N^T$.
Noting that $\xi=0$ if and only if $x_1=\cdots=x_N$, we call $\xi$ the consensus error, whose dynamics is given by
\begin{equation}\label{xid}
\begin {aligned}
\dot \xi&=(I_N\otimes A)\xi+(c\mathcal L_{\theta}\otimes BG)\tilde \xi.
\end{aligned}
\end{equation}

%
%

Note that the control law \dref{pro1} is only updated according to the information received at the latest event time instant, defined by
\begin{equation}\label{tk}
\begin {aligned}
t_{k+1}^i&\triangleq \mathrm{inf}\{t>t_k^i~|~f_i(t)\geq 0~\mathrm{or}~a_{ij}(t)\neq a_{ij}(t_k^i)\\
&~~~~~~~~\mathrm{for~some}~j\in \mathcal V\},
\end{aligned}
\end{equation}
where $t_0^i\triangleq 0$ and $f_i(t)$ is the triggering function defined as follows:
\begin{equation}\label{eve}
\begin {aligned}
f_i(t)&=4d_i(t)\|G\|^2\|e_i\|^2-\delta\sum_{j=1}^N{a_{ij}(t)\|G(\tilde x_i-\tilde x_j)\|^2}\\
&\quad-\mu e^{-\nu t},~i=1,\cdots,N,
\end{aligned}
\end{equation}
with $\delta$, $\mu$, $\nu$ being positive constants, and $e_i\triangleq\tilde x_i-x_i$ being the measurement error.
Once $f_i$ triggers, agent $i$ broadcasts its current state to neighbors.
The controllers \dref{pro1} of $i$ and its neighbors update immediately, and $e_i(t)$ resets at the same time.

\subsection{Event-Based Consensus Conditions}
Since $A$ is neutrally stable, in light of Lemmas 22 and 23 of \cite{ZLi2014cooperative}, we can choose $E\in \mathbf R^{{n_1}\times n}$ and $F\in \mathbf R^{{(n-n_1)}\times n}$ satisfying
\begin{equation*}\label{UW}
\begin {aligned}
\begin{bmatrix}E\\F\end{bmatrix}A\begin{bmatrix}E\\F\end{bmatrix}^{-1}=\begin{bmatrix}X&0\\0&Y\end{bmatrix},
\end{aligned}
\end{equation*}
where $X\in \mathbf R^{{n_1}\times {n_1}}$ is skew-symmetric and $Y\in \mathbf R^{{(n-n_1)}\times {(n-n_1)}}$ is Hurwitz.

\begin{remark}
It should be pointed out that the matrices $E$ and $F$ can be derived by rendering the matrix $A$ into the real Jordan canonical form \cite{r-a-horn-1990-matrix}.
\end{remark}

Choose $z=\left(I_N\otimes \begin{bmatrix}E\\F\end{bmatrix}\right)\xi$ and $\tilde z=\left(I_N\otimes \begin{bmatrix}E\\F\end{bmatrix}\right)\tilde \xi$.
The derivative of $z$ is given by
\begin{equation}\label{zd0}
\begin {aligned}
\dot z=\left(I_N\otimes \begin{bmatrix}X&0\\0&Y\end{bmatrix}\right)z+\left(c\mathcal L\otimes \begin{bmatrix}E\\F\end{bmatrix}BG\begin{bmatrix}E\\F\end{bmatrix}^{-1}\right)\tilde z.
\end{aligned}
\end{equation}
Let $H=EB$.
According to Assumption \ref{assumption_A}, $(X,H)$ is controllable.
Choose $E^+\in \mathbf R^{n\times {n_1}}$ and $F^+\in\mathbf R^{n\times {(n-n_1)}}$ satisfying $\begin{bmatrix}E^+&F^+\end{bmatrix}=\begin{bmatrix}E\\F\end{bmatrix}^{-1}$, with $EE^+=I$, $FF^+=I$, $FE^+=0$, and $EF^+=0$.
Letting $G=-B^TE^TE$, then we have
\begin{equation}\label{zd'}
\begin {aligned}
\dot z=\left(I_N\otimes \begin{bmatrix}X&0\\0&Y\end{bmatrix}\right)z-\left(c\mathcal L_{\theta}\otimes \begin{bmatrix}EBB^TE^T&0\\FBB^TE^T&0\end{bmatrix}\right)\tilde z.
\end{aligned}
\end{equation}
Define $z_I=(I_N\otimes E)\xi$, $\tilde z_I=(I_N\otimes E)\tilde \xi$, $z_{II}=(I_N\otimes F)\xi$ and $\tilde z_{II}=(I_N\otimes F)\tilde \xi$.
Rewrite \dref{zd'} as
\begin{subequations}\label{z12d}
\renewcommand{\theequation}
{\theparentequation-\arabic{equation}}
\begin{equation}\label{z1d}
\dot z_I=(I_N\otimes X)z_I-(c\mathcal L_{\theta}\otimes HH^T)\tilde z_I,\end{equation}
\begin{equation}\label{z2d}
\dot z_{II}=(I_N\otimes Y)z_{II}-(c\mathcal L_{\theta}\otimes FBB^TE^T)\tilde z_I.\end{equation}
\end{subequations}

\begin{lemma}(Cauchy's Convergence Criterion \cite{z-g-wu-event})\label{lemma-cauchys}
The sequence $V(\bar t_k)$, $k=0,1,2,\cdots$ converges if and only if
for $\forall\varepsilon>0$, $\exists M_\varepsilon\in\mathbf{Z}_+$ satisfying $\forall k>M_\varepsilon$,
$\left|V(\bar t_{k+1})-V(\bar t_k)\right|<\varepsilon$.
\end{lemma}

\begin{lemma}(Barbalat's Lemma \cite{PIoannou1996robust})\label{lemma_bar}
If $\lim_{t\rightarrow \infty}g(t)=a$ ($a$ is bounded) and $g''(t)$ is also bounded, then $\lim_{t\rightarrow \infty}g'(t)=0$.
\end{lemma}

Next, we introduce the main results of this section.

\begin{theorem}\label{theorem-1}
State consensus of the homogeneous subsystems \dref{sys_homo} is achieved under the event-driven algorithm composed of \dref{pro1} and \dref{eve} with $c>0$, $0<\delta<1$, $\mu>0$, $\nu>0$, and $G=-B^TE^TE$ \footnote{The matrix $E$ can be obtained according to Remark 1.}.
\end{theorem}

\begin{proof}
Let
\begin{equation}\label{lya0}
\begin {aligned}
V_1=\frac{1}{2}z_I^Tz_I.
\end{aligned}
\end{equation}
In light of \dref{z1d}, differentiating $V_1$ with respect to $t$ gives
\begin{equation}\label{lya1d}
\begin {aligned}
\dot V_1&=\frac{1}{2}z_I^T[I_N\otimes (X+X^T)]z_I-z_I^T(c\mathcal L_{\theta}\otimes HH^T)\tilde z_I.
\end{aligned}
\end{equation}
Since $X$ is skew-symmetric, $z_I^T[I_N\otimes (X+X^T)]z_I=0$.
Then, we have
\begin{equation}\label{chulia0}
\begin {aligned}
\dot V_1
&=-\frac{1}{2}z_I^T(c\mathcal L_{\theta}\otimes HH^T)z_I-\frac{1}{2}\tilde z_I^T(c\mathcal L_{\theta}\otimes HH^T)\tilde z_I\\
&\quad+\frac{1}{2}e^T(c\mathcal L_{\theta}\otimes G^TG)e.
\end{aligned}
\end{equation}


Let
\begin{equation}\label{Vz2}
\begin {aligned}
V_2=\frac{1}{2}z_{II}^T(I_N\otimes P)z_{II},
\end{aligned}
\end{equation}
where $P$ satisfies
\begin{equation}\label{ARE3}
\begin {aligned}
PY+Y^TP+2I=0.
\end{aligned}
\end{equation}
In light of \dref{z2d}, differentiating $V_2$ with respect to $t$ gives
\begin{equation}\label{Vz2d}
\begin {aligned}
\dot V_2&=\frac{1}{2}z_{II}^T[I_N\otimes (PY+Y^TP)]z_{II}\\
&\quad-z_{II}^T(c\mathcal L_{\theta}\otimes PFBB^TE^T)\tilde z_I.
\end{aligned}
\end{equation}
Using the Young's Inequality \cite{Nowzari17Event} gives
\begin{equation}\label{chulib0}
\begin {aligned}
&\quad-z_{II}^T(c\mathcal L_{\theta}\otimes PFBB^TE^T)\tilde z_I
\leq \frac{1}{2}z_{II}^Tz_{II}\\
&+\frac{c^2\lambda_N(\mathcal L_{\theta})}{2}\tilde z_I^T(\mathcal L_{\theta}\otimes EBB^TF^TPPFBB^TE^T)\tilde z_I\\
&\leq \frac{1}{2}z_{II}^Tz_{II}
+\frac{c\alpha_1}{2}\tilde x^T(\mathcal L_{\theta}\otimes G^TG)\tilde x,
\end{aligned}
\end{equation}
where $\alpha_1=c\lambda_N(\mathcal L)\|PFB\|^2$ and $\lambda_N(\mathcal L)$ denotes the largest eigenvalue of $\mathcal L_{\theta(t)}$ for all $t>0$.

Construct the Lyapunov function candidate as
\begin{equation}\label{V}
\begin {aligned}
V_3=\frac{\alpha_1}{1-\delta}V_1+V_2.
\end{aligned}
\end{equation}
Evidently, $V_3$ is positive definite,
whose derivative is given by
\begin{equation}\label{Vd2}
\begin {aligned}
\dot V_3
&\leq
\frac{c\alpha_1}{2(1-\delta)}[-z_I^T(\mathcal L_{\theta}\otimes HH^T)z_I\\
&\quad+e^T(\mathcal L_{\theta}\otimes G^TG)e
-\tilde x^T(\mathcal L_{\theta}\otimes G^TG)\tilde x]\\
&\quad+\frac{1}{2}z_{II}^T[I_N\otimes (PY+Y^TP+I)]z_{II}\\
&\quad+\frac{c\alpha_1}{2}\tilde x^T(\mathcal L_{\theta}\otimes G^TG)\tilde x\\
&\leq -\alpha_2z_I^T(\mathcal L_{\theta}\otimes HH^T)z_I-\frac{1}{2}z_{II}^Tz_{II}\\
&\quad+\alpha_2[e^T(\mathcal L_{\theta}\otimes G^TG)e-\delta\tilde x^T(\mathcal L_{\theta}\otimes G^TG)\tilde x],
\end{aligned}
\end{equation}
where $\alpha_2=\frac{c\alpha_1}{2(1-\delta)}$.
Because $a_{ij}(t)=a_{ji}(t)$, we have
\begin{equation}\label{chuli2}
\begin {aligned}
e^T(\mathcal L_{\theta}\otimes G^TG)e&=\sum_{i=1}^N\sum_{j=1}^N{a_{ij}(t)e_i^TG^TG(e_i-e_j)}\\
&\leq 2\sum_{i=1}^N\sum_{j=1}^N{a_{ij}(t)e_i^TG^TGe_i}\\
&\leq 2\sum_{i=1}^N{d_i(t)\|G\|^2\|e_i\|^2},
\end{aligned}
\end{equation}
and
\begin{equation}\label{chuli3}
\begin {aligned}
&\quad\tilde x^T(\mathcal L_{\theta}\otimes G^TG)\tilde x\\
&=\sum_{i=1}^N\sum_{j=1}^N{a_{ij}(t)\tilde x_i^TG^TG(\tilde x_i-\tilde x_j)}\\
&=\frac{1}{2}\sum_{i=1}^N\sum_{j=1}^N{a_{ij}(t)(\tilde x_i-\tilde x_j)^TG^TG(\tilde x_i-\tilde x_j)}\\
&=\frac{1}{2}\sum_{i=1}^N\sum_{j=1}^N{a_{ij}(t)\|G(\tilde x_i-\tilde x_j)\|^2}.
\end{aligned}
\end{equation}
By substituting \dref{eve}, \dref{chuli2}, and \dref{chuli3} into \dref{Vd2}, we have
\begin{equation}\label{Vd3}
\begin {aligned}
&\dot V_3\leq
-\alpha_2z_I^T(\mathcal L_{\theta}\otimes HH^T)z_I-\frac{1}{2}z_{II}^Tz_{II}\\
&+\frac{\alpha_2}{2}\sum_{i=1}^N\{4d_i(t)\|G\|^2\|e_i\|^2-\delta\sum_{j=1}^Na_{ij}(t)\|G(\tilde x_i-\tilde x_j)\|^2\}\\
&\leq -\alpha_2z_I^T(\mathcal L_{\theta}\otimes HH^T)z_I-\frac{1}{2}z_{II}^Tz_{II}+\frac{\mu \alpha_2N}{2}e^{-\nu t}.
\end{aligned}
\end{equation}

Define $\tilde V_3(t)=V_3(t)+\frac{\mu \alpha_2N}{2\nu}e^{-\nu t}$. Then, we have
\begin{equation}\label{tildeVd3}
\begin {aligned}
\dot{\tilde V}_3
&\leq -\alpha_2z_I^T(\mathcal L_{\theta}\otimes HH^T)z_I-\frac{1}{2}z_{II}^Tz_{II}.
\end{aligned}
\end{equation}
Combining with $\dot{\tilde V}_3(t)\leq 0$ and $\tilde V_3(t)\geq 0$, we have $\tilde V_3$ is bounded and $\lim_{t\rightarrow +\infty}\tilde V_3(t)$ exists.
Based on Lemma \ref{lemma-cauchys}, for $\forall \varepsilon>0$,
$\exists M_\varepsilon\in\mathbf{Z}_+$ satisfying $\forall k\geq M_\varepsilon$,
\begin{equation*}\label{cau0}
\left|\tilde V_3(\bar t_{k+1})-\tilde V_3(\bar t_k)\right|<\varepsilon,
\end{equation*}
or
\begin{equation*}\label{cauint0}
\left|\int_{\bar t_k}^{\bar t_{k+1}}\dot{\tilde V}_3(t)dt\right|<\varepsilon.
\end{equation*}
It follows that
\begin{equation}\label{cauint20}
\left|\int_{\bar t_k^0}^{\bar t_k^1}\dot{\tilde V}_3(t)dt\right|+\cdots+\left|\int_{\bar t_k^{m_k-1}}^{\bar t_k^{m_k}}\dot{\tilde V}_3(t)dt\right|<\varepsilon.
\end{equation}

In light of \dref{tildeVd3}, for each subinterval $[\bar t_k^j,\bar t_k^{j+1})$, $j=0,1,\cdots, m_k-1$, we have
\begin{equation}\label{each0}
\begin{aligned}
\left|\int_{\bar t_k^j}^{\bar t_k^{j+1}}\dot{\tilde V}_3(t)dt\right|
&\geq \alpha_2\int_{\bar t_k^j}^{\bar t_k^{j+1}}z_I^T(t)(\mathcal L_{\theta(\bar t_k^j)}\otimes HH^T)z_I(t)dt\\
&\quad+\frac{1}{2}\int_{\bar t_k^j}^{\bar t_k^{j+1}}z_{II}^T(t)z_{II}(t)dt\\
&\geq \alpha_2\int_{\bar t_k^j}^{\bar t_k^j+\tau}z_I^T(t)(\mathcal L_{\theta(\bar t_k^j)}\otimes HH^T)z_I(t)dt\\
&\quad+\frac{1}{2}\int_{\bar t_k^j}^{\bar t_k^j+\tau}z_{II}^T(t)z_{II}(t)dt.
\end{aligned}
\end{equation}
Combining \dref{cauint20} with \dref{each0} gives
\begin{equation*}\label{vare0}
\begin{aligned}
\varepsilon&>\alpha_2\bigg\{\int_{\bar t_k^0}^{\bar t_k^0+\tau}z_I^T(t)(\mathcal L_{\theta(\bar t_k^0)}\otimes HH^T)z_I(t)dt+\cdots\\
&\quad+\int_{\bar t_k^{m_k-1}}^{\bar t_k^{m_k-1}+\tau}z_I^T(t)(\mathcal L_{\theta(\bar t_k^{m_k-1})}\otimes HH^T)z_I(t)dt\bigg\},
\end{aligned}
\end{equation*}
which implies that for $\forall k>M_\varepsilon$,
\begin{equation}\label{vare20}
\begin{aligned}
\int_{\bar t_k^j}^{\bar t_k^j+\tau}z_I^T(t)&(\mathcal L_{\theta(\bar t_k^j)}\otimes HH^T)z_I(t)dt<\frac{\varepsilon}{\alpha_2},\\
&~j=0,1,\cdots,m_k-1.
\end{aligned}
\end{equation}
From \dref{vare20}, we have
\begin{equation*}\label{varelim0}
\begin{aligned}
\lim_{t\rightarrow \infty}\int_t^{t+\tau}z_I^T(s)&(\mathcal L_{\theta(\bar t_k^j)}\otimes HH^T)z_I(s)ds=0,\\
&~j=0,1,\cdots,m_k-1.
\end{aligned}
\end{equation*}
Since only finite switches take place during $[\bar t_k,\bar t_{k+1})$, we obtain that
\begin{equation*}
\begin{aligned}
&\lim_{t\rightarrow \infty}\int_t^{t+\tau}\bigg\{z_I^T(s)(\mathcal L_{\theta(\bar t_k^0)}\otimes HH^T)z_I(s)+\cdots\\
&\quad+z_I^T(s)(\mathcal L_{\theta(\bar t_k^{m_k-1})}\otimes HH^T)z_I(s)\bigg\}ds=0,
\end{aligned}
\end{equation*}
which can be rewritten as
\begin{equation}\label{limsum_homo}
\begin{aligned}
&\lim_{t\rightarrow \infty}\int_t^{t+\tau}\bigg\{z_I^T(s)(\mathcal L_\Sigma\otimes HH^T)z_I(s)\bigg\}ds=0,
\end{aligned}
\end{equation}
where $\mathcal L_\Sigma=\mathcal L_{\theta(\bar t_k^0)}+\cdots+\mathcal L_{\theta(\bar t_k^{m_k-1})}$.
According to Assumption \ref{assumption_gr}, $\mathcal L_\Sigma$ is connected.
We can find an orthogonal matrix $T_\Sigma$ such that $T_\Sigma\mathcal L_{\Sigma}T_\Sigma^T=\Lambda_\Sigma\triangleq \mathrm{diag}(0,\lambda_\Sigma^2,\cdots,\lambda_\Sigma^N)$, where $\lambda_\Sigma^i>0$, $i=2,\cdots,N$, are the eigenvalues of $\mathcal L_{\Sigma}$.
Define $\rho=[\rho_1^T,\cdots,\rho_N^T]^T=(T_{\Sigma}\otimes H^T)z_I$.
It is not difficult to verify that $\rho_1\equiv 0$.
Then, \dref{limsum_homo} implies that
\begin{equation*}\label{limsum20}
\begin{aligned}
\lim_{t\rightarrow \infty}\int_t^{t+\tau}\bigg\{\sum_{i=2}^N\lambda_\Sigma^i\rho_i^T(s)\rho_i(s)\bigg\}ds=0.
\end{aligned}
\end{equation*}
Because $\tilde V_3\geq 0$ is bounded and $0\leq V_3\leq \tilde V_3$, we conclude that $V_3$ is bounded.
In light of \dref{V}, $\rho(t)$ is bounded.
Noting that $\dot{\tilde z}=(I_N\times A)\tilde z$ and Assumption \ref{assumption_A}, we have $\tilde z_I$ is bounded. According to \dref{z1d}, we further get that $\dot{\rho}(t)$ is bounded.
Furthermore,
\begin{equation*}
\frac{d^2}{dt^2}\int_t^{t+\tau}\bigg\{\sum_{i=2}^N\lambda_\Sigma^i\rho_i^T(s)\rho_i(s)\bigg\}ds
=2\sum_{i=2}^N\lambda_\Sigma^i\rho_i^T(t)\dot{\rho}_i(t),
\end{equation*}
which is also bounded.
According to Lemma \ref{lemma_bar}, we have $\lim_{t\rightarrow \infty}\bigg\{\sum_{i=2}^N\lambda_\Sigma^i\rho_i^T(t)\rho_i(t)\bigg\}=0$, which further indicates that $\lim_{t\rightarrow \infty}\rho_i=0$, $\forall i\in\mathcal V$, i.e., $\lim_{t\rightarrow \infty}(I_N\otimes H^T)z_I(t)=0$.
Similarly, we can show that $\lim_{t\rightarrow \infty}z_{II}=0$.

In the following, we aim at showing that $\lim_{t\rightarrow \infty}z_I(t)=0$.

We first get from the triggering function \dref{eve} and the triggering rule that
\begin{equation*}
\begin{aligned}
&\quad e^T(\mathcal L_\theta\otimes G^TG)e\leq \delta\tilde x^T(\mathcal L_\theta\otimes G^TG)\tilde x+\frac{N\mu}{2}e^{-\nu t}\\
&\leq \frac{\delta}{1-\delta}x^T(\mathcal L_\theta\otimes G^TG)x+\frac{1+\delta}{2}e^T(\mathcal L_\theta\otimes G^TG)e\\
&\quad+\frac{N\mu}{2}e^{-\nu t},
\end{aligned}
\end{equation*}
where we have used the Young's inequality to get the last inequality.
Then, it follows that
$\frac{1-\delta}{2}e^T(\mathcal L_\theta\otimes G^TG)e\leq \frac{\delta}{1-\delta}z_I^T(\mathcal L_\theta\otimes HH^T)z_I+\frac{N\mu}{2}e^{-\nu t}$.
Since $\lim_{t\rightarrow \infty}(I_N\otimes H^T)z_I=0$, we further get that $\lim_{t\rightarrow \infty}e^T(\mathcal L_\theta\otimes G^TG)e=0$, which implies that $\lim_{t\rightarrow \infty}(\mathcal L_\theta\otimes HH^TE)e=0$.
We can rewrite \dref{z1d} as 
\begin{equation}\label{dotz1}
\dot z_I=(I_N\otimes X)z_I+\theta(t),
\end{equation}
where $\theta(t)=-(c\mathcal L_\theta\otimes HH^T)z_I-(c\mathcal L_\theta\otimes HH^TE)e$.
In light of the fact that $\lim_{t\rightarrow \infty}(I_N\otimes H^T)z_I=0$, shown as above, it is not difficult to find that $\lim_{t\rightarrow \infty}\theta(t)=0$.
According to \dref{dotz1}, we have
\begin{equation}\label{zi}
z_I(t)=e^{(I_N\otimes X)(t-\bar t_k)}z_I(\bar t_k)+\int_{\bar t_k}^te^{(I_N\otimes X)(t-r)}\theta(r)dr.
\end{equation}

We still consider $V_1=\frac{1}{2}z_I^Tz_I$ as in \dref{lya0} and by using the triggering function \dref{eve} can get that
\begin{equation}\label{5}
\dot V_1\leq -\frac{c}{2}z_I^T(\mathcal L_\theta\otimes HH^T)z_I+\frac{N\mu}{4}e^{-\nu t}.
\end{equation}
According to this, both $V_1$ and $z_I$ are always bounded. 
Considering a time interval $[\bar t_k,\bar t_{k+1}]$ and noting the switching rule of the topologies described in Section II-A, we have
\begin{equation*}
\begin{aligned}
&\quad V_1(\bar t_{k+1})-V_1(\bar t_k)=\int_{\bar t_k}^{\bar t_{k+1}}\dot V_1dt\\
&\leq -\frac{c}{2}\int_{\bar t_k}^{\bar t_{k+1}}z_I^T(\mathcal L_\theta\otimes HH^T)z_Idt+\frac{N\mu}{4}\int_{\bar t_k}^{\bar t_{k+1}}e^{-\nu t}dt\\
&=-\frac{c}{2}\Big[\int_{\bar t_k^0}^{\bar t_k^1}z_I^T(\mathcal L_{\theta(\bar t_k^0)}\otimes HH^T)z_Idt+\cdots\\
&\qquad+\int_{\bar t_k^{m_k-1}}^{\bar t_k^{m_k}}z_I^T(\mathcal L_{\theta(\bar t_k^{m_k-1})}\otimes HH^T)z_Idt\Big]+\beta_1\\
&\leq -\frac{c}{2}\int_{\bar t_k}^{\bar t_k+\tau}z_I^T(\mathcal L_\Sigma\otimes HH^T)z_Idt+\beta_1\\
&\leq -\frac{c}{2}\lambda_\Sigma^2\int_{\bar t_k}^{\bar t_k+\tau}z_I^T(I_N\otimes HH^T)z_Idt+\beta_1,
\end{aligned}
\end{equation*}
where $\tau$ is the dwelling time, $\lambda_\Sigma^2$ is the smallest nonzero eigenvalue of $\mathcal L_{\Sigma}$ defined in \dref{limsum_homo}, and $\beta_1=\beta_1(\bar t_k,\bar t_{k+1})\triangleq \frac{N\mu}{4\nu}(e^{-\nu \bar t_k}-e^{-\nu \bar t_{k+1}})$.
Obviously, $\lim_{\bar t_k,\bar t_{k+1}\rightarrow\infty}\beta_1=0$.

In light of \dref{zi}, we have
\begin{equation*}
\begin{aligned}
&\quad-\int_{\bar t_k}^{\bar t_k+\tau}z_I^T(I_N\otimes HH^T)z_Idt\\
&=-\int_{\bar t_k}^{\bar t_k+\tau}\Big[e^{(I_N\otimes X)(t-\bar t_k)}z_I(\bar t_k)+\int_{\bar t_k}^te^{(I_N\otimes X)(t-r)}\theta(r)dr\Big]^T\\
&\quad\cdot(I_N\otimes HH^T)\Big[e^{(I_N\otimes X)(t-\bar t_k)}z_I(\bar t_k)+\int_{\bar t_k}^te^{(I_N\otimes X)(t-r)}\theta(r)dr\Big]dt\\
&\leq -\frac{1}{2}\int_{\bar t_k}^{\bar t_k+\tau}\Big[e^{(I_N\otimes X)(t-\bar t_k)}z_I(\bar t_k)\Big]^T(I_N\otimes HH^T)\\
&\quad\cdot\Big[e^{(I_N\otimes X)(t-\bar t_k)}z_I(\bar t_k)\Big]dt
+\int_{\bar t_k}^{\bar t_k+\tau}\Big[\int_{\bar t_k}^te^{(I_N\otimes X)(t-r)}\theta(r)dr\Big]^T\\
&\quad\cdot(I_N\otimes HH^T)\Big[\int_{\bar t_k}^te^{(I_N\otimes X)(t-r)}\theta(r)dr\Big]dt\\
&\leq -\frac{1}{2}z_I^T(\bar t_k)Wz_I(\bar t_k)+\|HH^T\|\beta_2,
\end{aligned}
\end{equation*}
where $W\triangleq \int_{\bar t_k}^{\bar t_k+\tau}\Big[e^{(I_N\otimes X)(t-\bar t_k)}\Big]^T(I_N\otimes HH^T)\Big[e^{(I_N\otimes X)(t-\bar t_k)}\Big]dt$, $\beta_2=\int_{\bar t_k}^{\bar t_k+\tau}\Big[\int_{\bar t_k}^te^{(I_N\otimes X)(t-r)}\theta(r)dr\Big]^T\Big[\int_{\bar t_k}^te^{(I_N\otimes X)(t-r)}\theta(r)dr\Big]dt$, and to get the first inequality we have used the Young's inequality.
On one hand, we have shown that $(X,H)$ is controllable.
In other words, $(H^T,X)$ is observable, which implies that $W$ is positive definite. Without loss of generality, assume that there is a positive constant $s_1$ such that $W\geq s_1I$.
On the other hand, using the well-known Cauchy-Schwartz inequality \cite{Bernau} gives
\begin{equation*}
\begin{aligned}
\beta_2&\leq \int_{\bar t_k}^{\bar t_k+\tau}(t-\bar t_k)\int_{\bar t_k}^t\Big[e^{(I_N\otimes X)(t-r)}\theta(r)\Big]^T\\
&\quad\cdot \Big[e^{(I_N\otimes X)(t-r)}\theta(r)\Big]drdt\\
&=\int_{\bar t_k}^{\bar t_k+\tau}(t-\bar t_k)\int_{\bar t_k}^t\|\theta(r)\|^2drdt,
\end{aligned}
\end{equation*}
where to get the last equality we have used the fact that $X$ is skew-symmetric.
Since $\lim_{t\rightarrow \infty}\theta(t)=0$, for $\forall \epsilon>0$, there exists $\bar t>0$ such that for $\forall t\geq \bar t$, $\|\theta\|\leq \epsilon$.
Then, we have
$\beta_2\leq \epsilon^2\int_{\bar t_k}^{\bar t_k+\tau}(t-\bar t_k)\int_{\bar t_k}^tdrdt=\frac{1}{3}\epsilon^2\tau^3$ for $\forall \bar t_k\geq \bar t$, which further implies that $\lim_{\bar t_k\rightarrow \infty}\beta_2=0$.
Thus, it holds that
\begin{equation}\label{v1}
V_1(\bar t_{k+1})-V_1(\bar t_k)\leq -\frac{c}{2}s_1\lambda_\Sigma^2V_1(\bar t_k)+\beta_1+\beta_3,
\end{equation}
where $\beta_3=\frac{c}{6}\lambda_\Sigma^2\|HH^T\|\epsilon^2\tau^3$, in which $\lim_{\bar t_k\rightarrow \infty} \beta_3=0$.
Without loss of generality, we can find a constant $s_2\in(0,1)$ such that $s_2\leq \frac{c}{2}s_1\lambda_\Sigma^2$ and rewrite \dref{v1} as
\begin{equation}\label{v1+}
V_1(\bar t_{k+1})-V_1(\bar t_k)\leq -s_2V_1(\bar t_k)+\beta_1+\beta_3.
\end{equation}
Then, we can rewrite \dref{v1+} as
\begin{equation}\label{v1++}
V_1(\bar t_{k+1})\leq sV_1(\bar t_k)+\beta(\bar t_k).
\end{equation}
where $s=1-s_2\in(0,1)$ and $\beta(\bar t_k)=\beta_1+\beta_3$, in which $\lim_{\bar t_k\rightarrow \infty}\beta(\bar t_k)=\lim_{k\rightarrow \infty}\beta(\bar t_k)=0$.
Therefore, we have
\begin{equation*}
\begin{aligned}
V_1(\bar t_{k})&\leq sV_1(\bar t_{k-1})+\beta(\bar t_{k-1})\\
&\leq s^2V_1(\bar t_{k-2})+s\beta(\bar t_{k-2})+\beta(\bar t_{k-1})\\
&\quad\vdots\\
&\leq s^kV_1(\bar t_0)+s^{k-1}\beta(\bar t_0)+s^{k-2}\beta(\bar t_1)+\cdots\\
&\quad+s\beta(\bar t_{k-2})+\beta(\bar t_{k-1}).
\end{aligned}
\end{equation*}
Because $\bar t_0=0$, we further get that
\begin{equation}\label{8}
\begin{aligned}
V_1(\bar t_{k})&\leq s^kV_1(0)+s^{k-1}\beta(\bar t_0)+s^{k-2}\beta(\bar t_1)+\cdots\\
&\quad+s\beta(\bar t_{k-2})+\beta(\bar t_{k-1}).
\end{aligned}
\end{equation}
Since $s\in(0,1)$ and $\lim_{k\rightarrow \infty}\beta(\bar t_k)=0$, we must have $\lim_{\bar t_k\rightarrow \infty} V_1(\bar t_k)=0$ according to \dref{8}.

According to \dref{5}, for $\forall t\in[\bar t_k,\bar t_{k+1}]$, there exists that $V_1(\bar t_{k+1})+\beta_1(t,\bar t_{k+1})\leq V_1(t)\leq V_1(\bar t_k)+\beta_1(\bar t_k,t)$.
Noting that $\lim_{t\rightarrow \infty}\beta_1(\bar t_k,t)=0$, we have $\lim_{t\rightarrow \infty}V_1(t)=\lim_{\bar t_k\rightarrow \infty}V_1(\bar t_k)=0$,  implying that $\lim_{t\rightarrow \infty}z_I=0$.

Until now, we have proved the convergence of $z_I$. Consequently, state consensus is achieved.
$\hfill $$\blacksquare$
\end{proof}

\begin{remark}
It should be mentioned that the above derivations are partly inspired by the proofs of Theorem 8.5 in \cite{Khalil2002nonliear} and Proposition 1 in \cite{Tuna}.
In light of Remark 1, the feedback matrix $G$ is easy to determine such that the event-based protocol \dref{pro1} and \dref{eve} satisfies Theorem 1.
Contrary to \cite{z-g-wu-event,wyxu2017event}, where the designs of the event-based protocols rely on a solution to two coupled matrix inequalities, the existence of which is unclear in general cases,
the protocol proposed in this paper can be explicitly constructed, without the need to solve any matrix equality or inequality.
Besides, our protocol, requiring neither the switching rule of topologies nor nonzero eigenvalues of the Laplacian matrix,
can be devised and utilized in a completely distributed manner.
\end{remark}

%
%


\begin{theorem}\label{zeno}
The closed-loop system
\dref{xid} exhibits no Zeno behaviors and the interval between two consecutive triggering instants for any agent is strictly positive in finite time.
\end{theorem}

\begin{proof}
To exclude Zeno behaviors, we consider the following four cases.

i) In the first case, both $t_k^i$ and $t_{k+1}^i$ are determined by the triggering function \dref{eve}.
Under Assumption \ref{assumption_gr}, we only need to exclude Zeno behaviors for the network \dref{xid} when $d_i(t)>0$.
Combining with \dref{sys_homo} and \dref{pro1} gives
\begin{equation*}\label{edot}
\begin {aligned}
\dot e_i=Ae_i-c\sum_{j=1}^Na_{ij}(t)BG(\tilde x_i-\tilde x_j),
\end{aligned}
\end{equation*}
which implies that
\begin{equation}\label{eidotnorm}
\begin {aligned}
\frac{d\|e_i\|}{dt}\leq \|A\|\|e_i\|+c\sum_{j=1}^Na_{ij}(t)\|BG\|\|\tilde x_i-\tilde x_j\|.
\end{aligned}
\end{equation}

Theorem \ref{theorem-1} shows that $\xi$ is bounded.
Since $A$ is neutrally stable (by Assumption \ref{assumption_A}), it is easy to see that $\tilde \xi$ is also bounded.
Combing \dref{sys_homo} and \dref{pro1} gives $\dot x=(I_N\otimes A)x+(c\mathcal L_{\theta}\otimes BG)\tilde \xi$.
Thus, $x$ is bounded, which further indicates the boundedness of $\tilde x$.
Then, it follows from \dref{eidotnorm} that
\begin{equation}\label{eidotnorm2}
\begin {aligned}
\frac{d\|e_i\|}{dt}\leq \|A\|\|e_i\|+c\sigma_i,
\end{aligned}
\end{equation}
where $\sigma_i$ denotes the upper bound of $\sum_{j=1}^Na_{ij}(t)\|BG\|\|\tilde x_i-\tilde x_j\|$ for $t$ from $t_k^i$ to $t_{k+1}^i$.

Define a function $\psi: [0,\infty)\rightarrow \mathbf{R}_+$, satisfying
\begin{equation}\label{diffequ}
\begin {aligned}
\dot \psi=\|A\|\psi+c\sigma_i,~\psi(0)=\|e_i(t_k^i)\|=0.
\end{aligned}
\end{equation}
Then, we obtain that $\|e_i(t)\|\leq \psi(t-t_k^i)$, where $\psi(t)$ is the analytical solution to \dref{diffequ}, given by $\psi(t)=\frac{c\sigma_i}{\|A\|}\left(e^{\|A\|t}-1\right)$.

On the other hand, the triggering function \dref{eve} satisfies $f_i(t)\leq 0$, if we have the following condition:
\begin{equation}\label{condition}
\begin {aligned}
\|e_i\|^2\leq \frac{\mu e^{-\nu t}}{d_i(t)\|G\|^2}.
\end{aligned}
\end{equation}
Then, the interval between two triggering instants $t_k^i$ and $t_{k+1}^i$ for agent $v_i$ can be lower bounded by the time for $\psi^2(t-t_k^i)$ evolving from 0 to the right hand of \dref{condition}.
Thus, a lower bound of $t_{k+1}^i-t_k^i$, denoted as $\tau_k^i$, can be obtained by solving the following inequality
\begin{equation}
\frac{c^2\sigma_i^2}{\|A\|^2}\left(e^{\|A\|t}-1\right)^2\geq \frac{\mu e^{-\nu t}}{d_i(t)\|G\|^2},
\end{equation}
from which, we have
\begin{equation}\label{lowerbound}
\begin {aligned}
\tau_k^i\geq \frac{1}{\|A\|}\mathrm{ln}\left(1+\frac{\|A\|}{c\sigma_i\|G\|}\sqrt{\frac{\mu e^{-\nu (t_k^i+\tau_k^i)}}{d_i(t)}}\right).
\end{aligned}
\end{equation}

ii) In the second case, $t_k^i$ is determined by the switch of the topology, while $t_{k+1}^i$ is determined by the triggering function \dref{eve}.
Since the measurement error $e_i$ is reset to zero at $t_k^i$, this case is similar to the first case and the details are omitted here for brevity.

iii) In the third case, both $t_k^i$ and $t_{k+1}^i$ are determined by the switches of the topology.
It is obvious that the interval is not less than the dwelling time $\tau$.

iv) In the last case, $t_k^i$ is determined by the triggering function
\dref{eve}, while $t_{k+1}^i$ is determined by the switch of the topology.
Note that in finite time, there is only a finite number of switches.
Therefore, the minimum of the finite interval $\tau_k^i=t_{k+1}^i-t_k^i$ is nonzero, and there exists a minimum inter-event time, while its value is not available in this case.

In conclusion, Zeno behaviors are excluded and the interval between two consecutive triggering instants is strictly positive in finite time.
$\hfill $$\blacksquare$
\end{proof}

\begin{remark}
Generally speaking, the Zeno behavior is excluded if there does not exist infinite triggers within a finite period of time.
However, as pointed out in \cite{Nowzari17Event}, even though the Zeno behavior is ruled out theoretically,
it is still troublesome from an implementation viewpoint, if the physical hardware cannot match the speed of actions required by the protocol.
In other words, ensuring a system does not exist the Zeno behavior may not be enough to guarantee the protocol can be implemented on a physical system.
As an expected feature, the triggering rule \dref{eve} designed in this paper guarantees that the interval between different triggering instants in finite time is not less than a strictly positive constant.
Besides, the hybrid triggering functions \dref{eve} including the state term $-\delta \sum_{j=1}^Na_{ij}(t)\|G(\tilde x_i-\tilde x_j)\|^2$ and the time term $-\mu e^{-\nu t}$ are more propitious to reduce communication frequency compared to the ones in \cite{DYang2016Decentralized} when the time $t$ becomes very long or even as $t\rightarrow \infty$.
\end{remark}

\begin{remark}
Theorems 1 and 2 show that the presented event-triggered algorithm is applicable to switching networks satisfying the jointly connected condition.
According to the triggering rule \dref{tk}, communications only take place when the triggering function \dref{eve} is violated or the topology switches.
It should be noted when $\tau\rightarrow +\infty$, the event-based protocol here is reduced to the one for fixed graphs as a special case.
If $\tau$ is too small, there is no need to check whether the triggering function \dref{eve} is violated or not and  communications is not required until the next switch of the topologies takes place.
\end{remark}

\section{event-based output consensus of heterogeneous multi-agent systems}\label{section_output}

\subsection{Problem Formulation}
In this section, we consider $N$ heterogeneous linear agents, whose dynamics can be described by
\begin{equation}\label{sys}
\begin{aligned}
\dot x_i&=A_ix_i+B_iu_i+E_iw_0,\\
y_i&=C_ix_i+F_iw_0,~i=1,\cdots,N,
\end{aligned}
\end{equation}
where $x_i\in \mathbf{R}^{n_i}$ denotes the state, $u_i\in \mathbf{R}^{m_i}$ represents the control input, $y_i\in\mathbf{R}^{p_i}$ is the output, and $A_i\in \mathbf{R}^{n_i\times n_i}$, $B_i\in \mathbf{R}^{n_i\times m_i}$, $C_i\in \mathbf{R}^{p_i\times n_i}$, $E_i\in \mathbf{R}^{n_i\times q}$, and $F_i\in \mathbf{R}^{p_i\times q}$ are constant matrices.
The exogenous signal $w_0\in \mathbf{R}^q$, which can be treated as a reference input or an external disturbance, satisfies the following dynamics:
\begin{equation}\label{exo}
\begin{aligned}
\dot w_0&=Sw_0,
\end{aligned}
\end{equation}
where
$S\in\mathbf{R}^{q\times q}$.

The objective here is to design distributed event-based algorithms under which all subsystems described by \dref{sys} converge to a common output
and Zeno behaviors can be eliminated.

Similarly as in \cite{Su2012output}, we can view the exosystem \dref{exo} as a leader, indexed by 0, and the $N$ subsystems \dref{sys} as followers, indexed by $1, \cdots,N$.
Denote $\Delta_{\theta}\triangleq \mathrm{diag}\{a_{10}(t),\cdots,a_{N0}(t)\}$, where $a_{i0}(t)=1$ if the leader is a neighbor of $i$ currently and $a_{i0}(t)=0$ otherwise.
Use $\mathcal{\bar G}_{\theta}$ to denote the leader-follower graph and
let $\mathcal H_{\theta}=\mathcal L_{\theta}+\Delta_{\theta}$.
The leader has directed pathes to all followers during $[\bar t_k,\bar t_{k+1})$, if the union graph $\bigcup_{t\in[\bar t_k,\bar t_{k+1})}\bar{\mathcal G}_{\theta(t)}$ contains a directed spanning tree with the leader as the root node.

\begin{assumption}\label{assumption_AB}
The pairs $(A_i,B_i)$, $\forall i\in\mathcal V$, are stabilizable.
\end{assumption}


\begin{assumption}\label{assumption_S}
$S$ has no eigenvalues with positive real parts.
\end{assumption}

\begin{assumption}\label{assumption_spec}
For all $\lambda\in \sigma(S)$, where $\sigma(S)$ represents the spectrum of $S$,
$\text{rank}\left(\begin{bmatrix}A_i-\lambda I&B_i\\C_i&0\end{bmatrix}\right)=n_i+p_i$.
\end{assumption}

\begin{assumption}\label{assumption_SR}
There exist solutions $R\in \mathbf{R}^{p_i\times q}$ such that the following regulator equations have solutions $\Pi_i\in \mathbf R^{n_i\times q}$ and $U_i\in \mathbf R^{m_i\times q}$:
\begin{equation}\label{regulator equa}
\begin{aligned}
\Pi_iS&=A_i\Pi_i+B_iU_i+E_i,\\
R&=C_i\Pi_i+F_i,~i=1,\cdots,N.
\end{aligned}
\end{equation}
\end{assumption}

\begin{assumption}\label{assumption_graph}
The leader has directed pathes to all followers in the union graph $\bigcup_{t\in[\bar t_k,\bar t_{k+1})}\bar{\mathcal G}_{\theta(t)}$.
\end{assumption}

\begin{remark}
Assumptions 3-6 are often used in the output consensus or regulation control of heterogeneous networks \cite{WHu2017output,WHu2017cooperative,YSu2012cooperative,JHuang2004nonliear}.
According to Assumption \ref{assumption_spec}, the transmission zeros of the system \dref{sys} do not coincide with the eigenvalues of the matrix $S$, which is often called the transmission zeros condition \cite{JHuang2004nonliear}.
Assumption \ref{assumption_SR} gives a characterization of the control objective in terms of the solvability of a set of linear matrix equations. This characterization allows the linear output consensus problem to be studied using the familiar mathematic tool of linear algebra.
\end{remark}

\subsection{Event-Based Estimates of the Exogenous Signal}

Since the exogenous signal \dref{exo} is available to only a subset of followers, we first design a distributed event-based observer for each follower as
\begin{equation}\label{obser}
\begin {aligned}
\dot w_i=Sw_i+c\sum_{j=0}^N{a_{ij}(t)(\tilde w_i-\tilde w_j)},~\forall i\in\mathcal V,
\end{aligned}
\end{equation}
where $c>0$, $w_i(t)$ represents the estimate of the exogenous signal $w_0(t)$, and $\tilde w_i(t)=e^{S(t-t_k^i)}w_i(t_k^i)$.
Denote $z_i=w_i-w_0$ and $\tilde z_i=\tilde w_i-\tilde w_0$, $i=1,\cdots,N$.
Let $z=[z_1^T,\cdots,z_N^T]^T$ and $\tilde z=[\tilde z_1^T,\cdots,\tilde z_N^T]^T$.
Let $z_0=0$ and $\tilde z_0=0$.
Then, it follows that $z=0$ if and only if $w_0=w_1=\cdots=w_N$.
Thus, $z_i$ satisfies the following dynamics:
\begin{equation}\label{zd}
\begin {aligned}
\dot z_i&=Sz_i-c\sum_{j=0}^N{a_{ij}(t)(\tilde z_i-\tilde z_j)},~\forall i\in\mathcal V.
\end{aligned}
\end{equation}
Rewrite \dref{zd} as
\begin{equation}\label{zd2}
\begin {aligned}
\dot z(t)&=(I_N\otimes S)z-(c\mathcal H_{\theta}\otimes I_q)\tilde z.
\end{aligned}
\end{equation}

Let $\varphi=[\varphi_1^T,\cdots,\varphi_N^T]^T=(I_N\otimes e^{-St})z$ and $\tilde \varphi=[\tilde\varphi_1^T,\cdots,\tilde\varphi_N^T]^T=(I_N\otimes e^{-St})\tilde z$ with $\varphi(0)=z(0)$ and $\tilde \varphi(0)=\tilde z(0)$.
It then follows from \dref{zd2} that
\begin{equation}\label{varphid}
\begin {aligned}
\dot \varphi&=-(I_N\otimes Se^{-St})z+(I_N\otimes e^{-St})\dot z\\
&=-(c\mathcal H_{\theta}\otimes I_q)\tilde \varphi.
\end{aligned}
\end{equation}

\begin{lemma}\label{lemma_var}
If $\varphi(t)$ converges to 0 exponentially, so does $z(t)$.
\end{lemma}
\begin{proof}
Based on the convergency of $\varphi$, we can choose constants $\mu_1$ and $\mu_2$ such that
\begin{equation*}
\begin{aligned}
\|\varphi(t)\|\leq \mu_1\|\varphi(0)\|e^{-\mu_2t}.
\end{aligned}
\end{equation*}
According to Assumption \ref{assumption_S}, there exists a polynomial $\Omega(t)$ satisfying
\begin{equation*}
\begin{aligned}
\|(I_N\otimes e^{St})\|\leq \Omega(t).
\end{aligned}
\end{equation*}
Since $\varphi=(I_N\otimes e^{-St})z$, we get
\begin{equation*}
\begin{aligned}
\|z(t)\|&\leq \|(I_N\otimes e^{St})\|\cdot \|\varphi(t)\|
\leq \mu_1\|z(0)\|\Omega(t)e^{-\mu_2t}.
\end{aligned}
\end{equation*}
This means if $\varphi(t)$ converges to 0 exponentially, so does $z(t)$.
$\hfill $$\blacksquare$
\end{proof}


Define the measurement error as
\begin{equation}\label{ei}
\begin {aligned}
e_i\triangleq \tilde w_i-w_i,~i=1,\cdots,N.
\end{aligned}
\end{equation}
Let $\epsilon=[\epsilon_1^T,\cdots,\epsilon_N^T]^T$ with $\epsilon_i\triangleq e^{-St}e_i(t)$, $i=1,\cdots,N$.
Event triggering instants are determined by \dref{tk} where
\begin{equation}\label{eve1}
\begin {aligned}
f_i(t)&=d_i(t)\|\epsilon_i\|^2
-\frac{1}{4}\sum_{j=0}^N{a_{ij}(t)\|\tilde w_i-\tilde w_j\|^2}-\mu e^{-\nu t},
\end{aligned}
\end{equation}
with $\tilde w_0\triangleq w_0$ and $d_i(t)$ being the degree of agent $i$ associated with the subgraph $\mathcal G_{\theta(t)}$.


\begin{theorem}\label{theorem_w}
The observers \dref{obser} with $c>0$ can track the exogenous signal $w_0(t)$ under the triggering function \dref{eve1}.
Moreover, there does not exist the Zeno behavior.
\end{theorem}

\begin{proof}
Construct the Lyapunov function candidate as
\begin{equation}\label{lya1}
\begin {aligned}
V_4&=\frac{1}{2}\varphi^T\varphi.
\end{aligned}
\end{equation}
Evidently, $V_4$ is positive definite, whose derivative is given by
\begin{equation}\label{lya1d1}
\begin {aligned}
\dot V_4
&=-\varphi^T(c\mathcal H_{\theta}\otimes I_q)\tilde \varphi\\
&=-c\sum_{i=1}^N{a_{i0}(t)\varphi_i^T\tilde \varphi_i}
-c\sum_{i=1}^N\sum_{j=1}^N{a_{ij}(t)\varphi_i^T(\tilde \varphi_i-\tilde \varphi_j)}.
\end{aligned}
\end{equation}
It is easy to verify that
\begin{equation}\label{lya3d2}
\begin {aligned}
-\sum_{i=1}^N&{a_{i0}(t)\varphi_i^T\tilde \varphi_i}
=-\frac{1}{2}\sum_{i=1}^N{a_{i0}(t)\varphi_i^T\varphi_i}\\
&\quad-\frac{1}{2}\sum_{i=1}^N{a_{i0}(t)\tilde \varphi_i^T\tilde \varphi_i}
+\frac{1}{2}\sum_{i=1}^N{a_{i0}(t)\epsilon_i^T\epsilon_i},
\end{aligned}
\end{equation}
and
\begin{equation}\label{lya3d4}
\begin {aligned}
&\quad-\sum_{i=1}^N{\sum_{j=1}^N{a_{ij}(t)\varphi_i^T(\tilde \varphi_i-\tilde \varphi_j)}}\\
&=-\frac{1}{4}\sum_{i=1}^N{\sum_{j=1}^N{a_{ij}(t)(\varphi_i-\varphi_j)^T(\varphi_i-\varphi_j)}}\\
&\quad-\frac{1}{4}\sum_{i=1}^N{\sum_{j=1}^N{a_{ij}(t)(\tilde \varphi_i-\tilde \varphi_j)^T(\tilde \varphi_i-\tilde \varphi_j)}}\\
&\quad+\frac{1}{4}\sum_{i=1}^N{\sum_{j=1}^N{a_{ij}(t)(\epsilon_i-\epsilon_j)^T (\epsilon_i-\epsilon_j)}}.
\end{aligned}
\end{equation}
Using the Young's Inequality gives
\begin{equation}\label{young1'}
\begin {aligned}
&\quad\sum_{i=1}^N\sum_{j=1}^Na_{ij}(t){(\epsilon_i-\epsilon_j)^T(\epsilon_i-\epsilon_j)}\\
&\leq 2\sum_{i=1}^N\sum_{j=1}^Na_{ij}(t){\epsilon_i^T\epsilon_i}
+2\sum_{i=1}^N\sum_{j=1}^Na_{ij}(t){\epsilon_j^T\epsilon_j}\\
&=4\sum_{i=1}^N\sum_{j=1}^Na_{ij}(t){\epsilon_i^T\epsilon_i}.
\end{aligned}
\end{equation}

Denote $\tilde\varphi_0=0$. Substituting \dref{zd}, \dref{lya3d2}, \dref{lya3d4}, and \dref{young1'} into \dref{lya1d1} yields
\begin{equation}\label{lya3d6}
\begin {aligned}
\dot V_4
&\leq-\frac{c}{2}\varphi^T(\mathcal H_{\theta}\otimes I_q)\varphi-\frac{c}{2}\sum_{i=1}^N{a_{i0}(t)\|\tilde \varphi_i\|^2}\\
&+\frac{c}{2}\sum_{i=1}^N{a_{i0}(t)\|\epsilon_i\|^2}-\frac{c}{4}\sum_{i=1}^N{\sum_{j=1}^N{a_{ij}(t)\|\tilde \varphi_i-\tilde \varphi_j\|^2}}\\
&\quad+c\sum_{i=1}^N{\sum_{j=1}^N{a_{ij}(t)\|\epsilon_i\|^2}}\\
&\leq -\frac{c}{2}\varphi^T(\mathcal H_{\theta}\otimes I_q)\varphi\\
&\quad+c\sum_{i=1}^N\left\{{d_i(t)\|\epsilon_i\|^2-\frac{1}{4}{\sum_{j=0}^N{a_{ij}(t)\|\tilde \varphi_i-\tilde \varphi_j\|^2}}}\right\}\\
&\leq -\frac{c}{2}\varphi^T(\mathcal H_{\theta}\otimes I_q)\varphi+c\mu Ne^{-\nu t},
\end{aligned}
\end{equation}
where we have used the triggering function \dref{eve1} to get the last inequality.


Similarly as in the proof of Theorem \ref{theorem-1}, we can prove that $\lim_{t\rightarrow \infty}\varphi(t)=0$.
According to Lemma \ref{lemma_var}, the observers \dref{obser} can track the exogenous signal $w_0(t)$.

Zeno behaviors can be similarly eliminated as in proof of Theorem \ref{zeno}.
$\hfill $$\blacksquare$
\end{proof}

\subsection{Distributed Control Inputs}
Upon the basis of the designed observer \dref{obser}, we present the following controller
\begin{equation}\label{controller}
\begin {aligned}
u_i&=K_{1i}x_i+K_{2i}w_i,~i=1,\cdots,N,
\end{aligned}
\end{equation}
where $w_i$ is defined in \dref{obser}, and $K_{1i}$ and $K_{2i}$ are feedback matrices to be designed.
Substituting \dref{controller} into \dref{sys} gives the following closed-loop dynamics:
\begin{equation}\label{closed}
\begin{aligned}
\dot x_i&=(A_i+B_iK_{1i})x_i+B_iK_{2i}w_i+E_iw_0,\\
y_i&=C_ix_i+F_iw_0,~i=1,\cdots,N.
\end{aligned}
\end{equation}


\begin{theorem}\label{theorem_out}
Select $K_{1i}$ such that $A_i+B_iK_{1i}$ are Hurwitz and $K_{2i}=U_i-K_{1i}\Pi_i$, $i=1,\cdots,N$, where $(X_i,\Pi_i,R_i)$ are unique solutions to regulator equations \dref{regulator equa}.
Output consensus is achieved under the event-based observer \dref{obser}, the triggering function \dref{eve1}, and the local controller \dref{controller}.
\end{theorem}

\begin{proof}
Let $\phi_i=x_i-\Pi_iw_0$. Noting that \dref{regulator equa} and $K_{2i}=U_i-K_{1i}\Pi_i$, we can rewrite \dref{closed} as
\begin{equation}\label{phid}
\begin{aligned}
\dot \phi_i&=(A_i+B_iK_{1i})\phi_i-(E_i-\Pi_iS)z_i,\\
y_i&=C_ix_i+F_iw_0,~i=1,\cdots,N.
\end{aligned}
\end{equation}
According to Theorem \ref{theorem_w}, we have $\lim_{t\rightarrow \infty}{z_i(t)}=0$.
Thus, if we choose $K_{1i}$, $i=1,\cdots,N$, such that $A_i+B_iK_{1i}$ are Hurwitz, it is not difficult to obtain the result that
$\lim_{t\rightarrow \infty}{\phi_i(t)}=0$, which further leads to
\begin{equation*}
\begin{aligned}
&\quad\lim_{t\rightarrow \infty}{(y_i(t)-y_j(t))}\\
&=\lim_{t\rightarrow \infty}[(C_ix_i+F_iw_0)-(C_jx_j+F_jw_0)]\\
&=\lim_{t\rightarrow \infty}{[(C_i\Pi_i+F_i)w_0-(C_j\Pi_j+F_j)w_0]}\\
&\quad+\lim_{t\rightarrow \infty}C_i\phi_i-\lim_{t\rightarrow \infty}C_j\phi_j\\
&=\lim_{t\rightarrow \infty}(R-R)w_0\\
&=0.
\end{aligned}
\end{equation*}
In conclusion, output consensus of heterogeneous systems \dref{sys} is achieved.
$\hfill $$\blacksquare$
\end{proof}

\begin{remark}
Theorems 3 and 4 show that the proposed protocol \dref{obser}, \dref{eve1}, and \dref{controller} is able to solve the event-driven output consensus control problem of heterogeneous networks.
In particular, the state consensus of homogeneous agents considered in Section III can be treated as a special case here, if we let $A_i = A$, $B_i = B$, $C_i = I$, $E_i=0$, and $F_i=0$, $\forall i\in\mathcal V$.
\end{remark}

\begin{remark}
Compared to \cite{Su2012output}, where output consensus of heterogeneous networks with continuous communications is considered,
the event-based protocol given in this paper does not require continuous communications either between sensors and controllers or among neighboring agents.
For each agent, both the control input and the triggering function are only based on state estimates of neighboring agents $\tilde w_j$ (or $\tilde x_j$) but not their real state $w_j$ (or $x_j$).
As for $e_i(t)=\tilde w_i(t)-w_i(t)$ (or $e_i(t)=\tilde x_i(t)-x_i(t)$), it can be computed according to its own information rather than neighbors' one.
In other words, discrete information of neighbors at event instants rather than continuous one is required for control laws' updating and triggering functions' monitoring.
Thus, the event-based protocols proposed in this paper are able to reduce communication frequency when implemented on practical systems.
\end{remark}

\section{Simulation Examples}\label{s_sim}
In this section, numerical simulations are introduced to demonstrate the effectiveness of the presented algorithms.

{\bf Example 1:}
The dynamics are described by \dref{sys_homo} with
$A=\begin{bmatrix}0&1\\-1&0\end{bmatrix}$ and $B=\begin{bmatrix}0&1\end{bmatrix}^T$.
All initial values of the agents are randomly chosen.
Denote the network graph as $\mathcal {G}_{\theta}$ with possible interaction graphs $\{\mathcal {G}_1,\mathcal {G}_2,\mathcal {G}_3,\mathcal {G}_4\}$ shown in Fig. \ref{fig:subfig}.
Note that there exist no any connections among these nodes in $\mathcal G_2$.
The interaction graphs are switched as $\mathcal {G}_1\rightarrow \mathcal {G}_2 \rightarrow \mathcal {G}_3 \rightarrow \mathcal {G}_4 \rightarrow \mathcal {G}_1\rightarrow \cdots$, and each graph is active for $0.5$ s.
The union graph associated with the agents is given in Fig. \ref{uniongraph}, which is connected,
implying that Assumption \ref{assumption_gr} holds \cite{BCheng2018CCTA}.

\begin{figure}
  \centering
  \subfigure[$\mathcal {G}_1$]{
\includegraphics[width=0.15\textheight,height=0.15\textwidth]{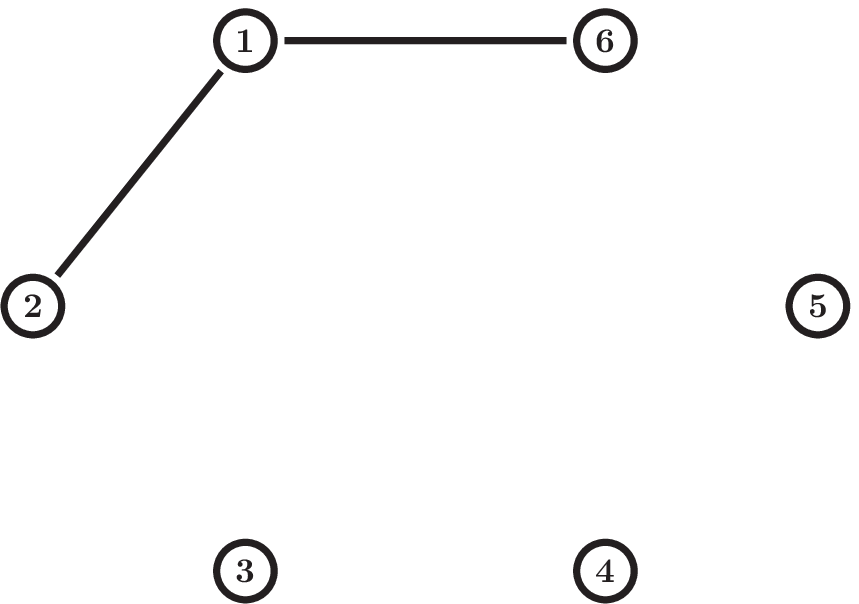}
    \label{fig:subfig:a} 
}
\quad \quad
  \subfigure[$\mathcal {G}_2$]{
\includegraphics[width=0.15\textheight,height=0.15\textwidth]{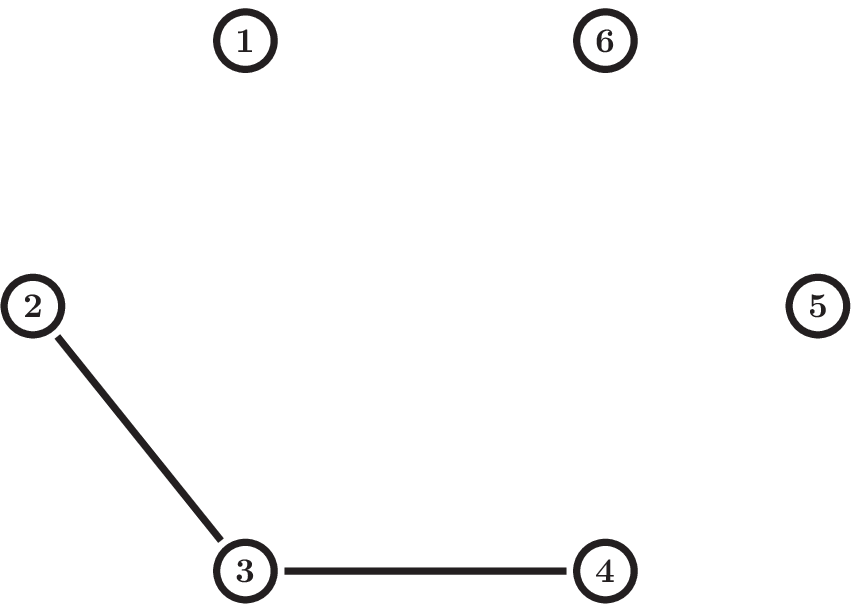}

    \label{fig:subfig:b} 
}
\newline
\newline
      \subfigure[$\mathcal {G}_3$]{
\includegraphics[width=0.15\textheight,height=0.15\textwidth]{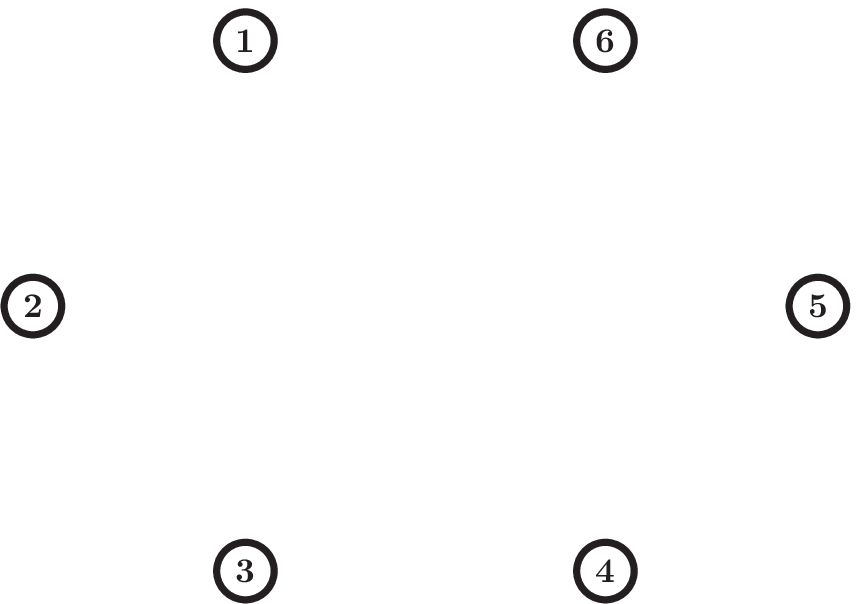}

    \label{fig:subfig:c} 
}
\quad \quad
  \subfigure[$\mathcal {G}_4$]{
\includegraphics[width=0.15\textheight,height=0.15\textwidth]{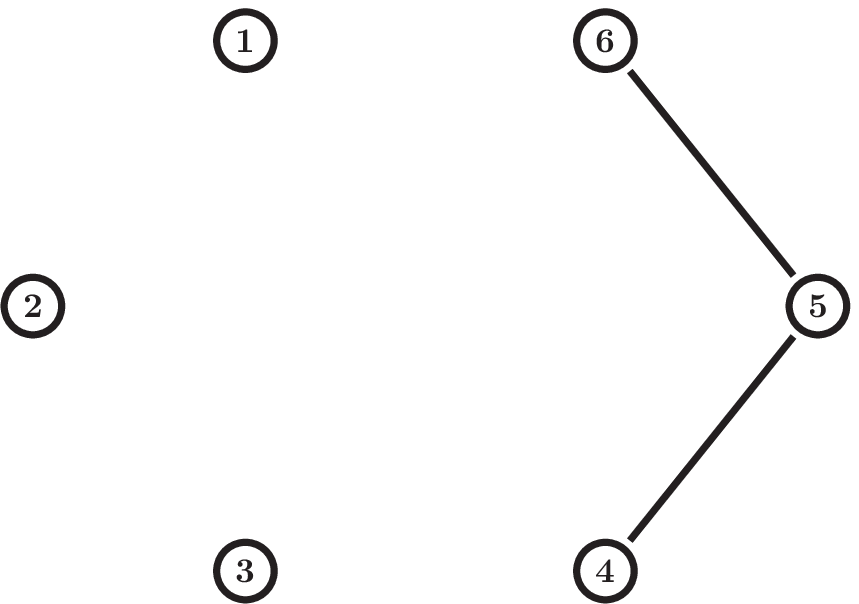}

    \label{fig:subfig:d} 
}
  \caption{Possible interaction topologies among the agents.}
  \label{fig:subfig} 
\end{figure}
\begin{figure}[!htb]
\centering
\includegraphics[width=0.25\textheight,height=0.25\textwidth]{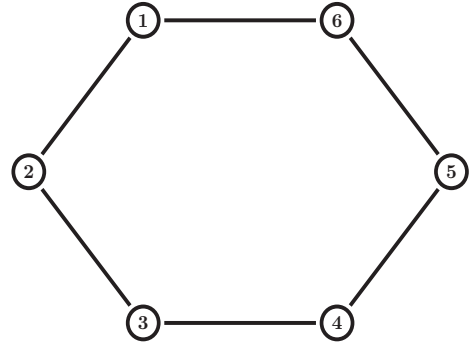}
\caption{The union graph associated with the agents.}
\label{uniongraph}
\end{figure}

To solve the consensus control problem, we use the event-based protocol \dref{pro1} and \dref{eve} with parameters chosen as
$c=5$, $\delta=\mu=\nu=0.5$, and $G=\begin{bmatrix}0&-1\end{bmatrix}$.
The states $x_i$, $i=1,\cdots,6$, are depicted in Fig. \ref{fig-consensuserror-1}, implying that consensus is achieved.
The triggering instants of all agents are presented in Fig. \ref{fig-triggering-1}, which shows that Zeno behaviors are ruled out.
\begin{figure}[!htb]
\centering
\includegraphics[width=0.4\textheight,height=0.4\textwidth]{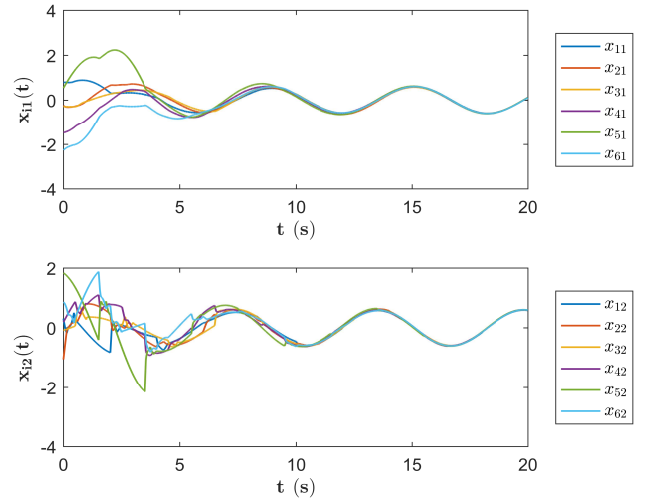}
\caption{The state of each agent.}
\label{fig-consensuserror-1}
\end{figure}
\begin{figure}[!htb]
\centering
\includegraphics[width=0.4\textheight,height=0.3\textwidth]{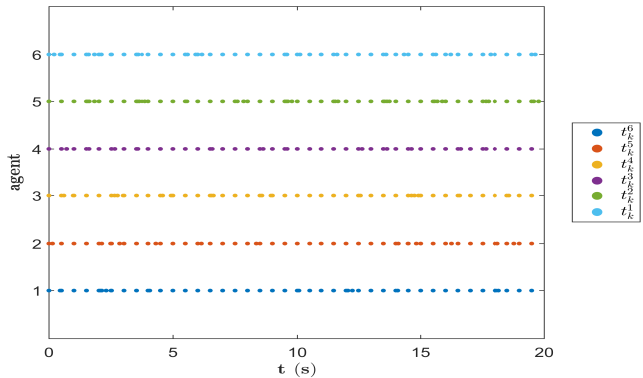}
\caption{Triggering instants of each agent.}
\label{fig-triggering-1}
\end{figure}

{\bf Example 2:}
The leader's dynamics satisfies \dref{exo} with $S=\begin{bmatrix}0&1\\-1&0\end{bmatrix}$
and the dynamics of followers are described by \dref{sys} with
$A_i=\begin{bmatrix}-1&1\\0&-i\end{bmatrix}$, $B_i=\begin{bmatrix}0\\i\end{bmatrix}$,
$C_i=\begin{bmatrix}i\\0\end{bmatrix}$,
$E_i=I_2$,
$F_i=\begin{bmatrix}1\\0\end{bmatrix}$,
$i=1,\cdots,4$.
All agents' initial values are randomly chosen.
Suppose that possible interaction topologies shown in Fig. \ref{fig:subfig} switches as $\mathcal {\bar G}_1\rightarrow \mathcal {\bar G}_2 \rightarrow \mathcal {\bar G}_3 \rightarrow \mathcal {\bar G}_4 \rightarrow \mathcal {\bar G}_1\rightarrow \cdots$, with the dwelling time $\tau=0.5$s.
Note that node $0$ represents the leader and nodes 1-4 denote followers.
It is not difficult to find that
Assumptions 3-7 are satisfied.
\begin{figure}[!htb]
\centering
\includegraphics[width=0.35\textheight,height=0.45\textwidth]{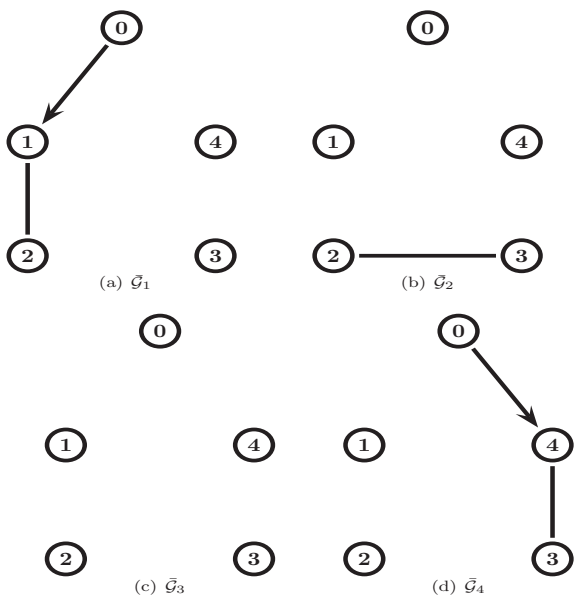}
\caption{Possible interaction topologies.}
\label{fig:subfig}
\end{figure}

To achieve output consensus, we utilize the event-triggered protocol \dref{obser}, \dref{eve1}, and \dref{controller}.
Solving the regulation equation \dref{regulator equa} gives
$\Pi_i=\begin{bmatrix}1/i&1/i\\-1&2/i\end{bmatrix}$,
$U_i=\begin{bmatrix}-1-2/i^2&0\end{bmatrix}$,
and $R=\begin{bmatrix}2&1\end{bmatrix}$,
$i=1,\cdots,4$.
Other parameters in this protocol are chosen as
$c=2$,
$K_{1i}=\begin{bmatrix}-1&-1\end{bmatrix}$,
and $K_{2i}=\begin{bmatrix}-2-1/i-2/i^2&3/i\end{bmatrix}$,
$i=1,\cdots,4$.

The estimate errors $w_i-w_0$, $i=1,\cdots,4$, for $t$ from $0s$ to $30s$, are depicted in Fig. \ref{fig-consensuserror}, implying that the observers \dref{obser} can track the exogenous signal $w_0(t)$.
Event instants of all followers are shown in Fig. \ref{fig-triggering}, indicating that there exist no Zeno behaviors.
The output errors $y_i-y_1$, $i=2,3,4$, are depicted in Fig. \ref{fig-output_error}, implying the achievement of output consensus.

\begin{figure}[!htb]
\centering
\includegraphics[width=0.35\textheight,height=0.4\textwidth]{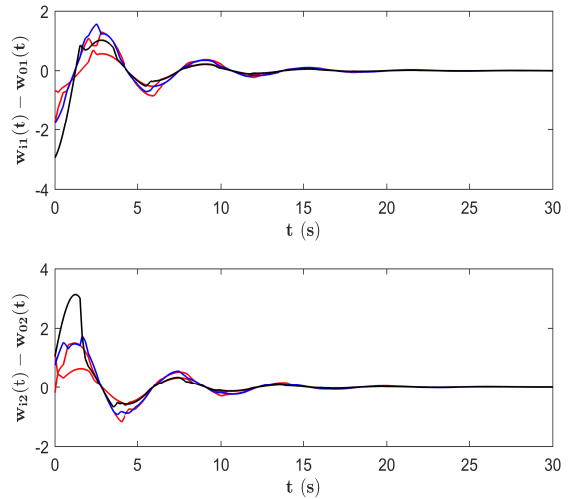}
\caption{The estimate errors $w_i-w_0$, $i=1,\cdots,4$.}
\label{fig-consensuserror}
\end{figure}

\begin{figure}[!htb]
\centering
\includegraphics[width=0.32\textheight,height=0.25\textwidth]{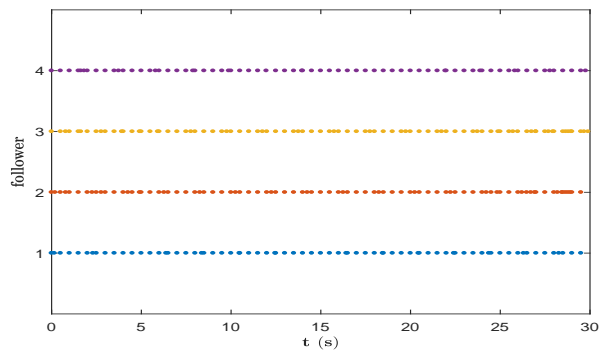}
\caption{Triggering instants of each follower.}
\label{fig-triggering}
\end{figure}
\begin{figure}[!htb]
\centering
\includegraphics[width=0.32\textheight,height=0.3\textwidth]{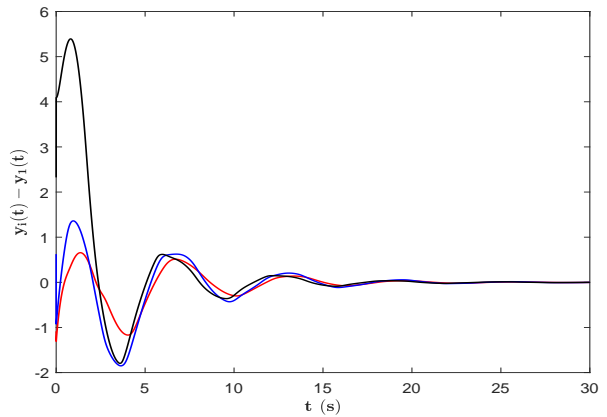}
\caption{The output errors $y_i-y_1$, $i=2,3,4$.}
\label{fig-output_error}
\end{figure}

\section{Conclusion}\label{s_con}
In this paper, distributed event-driven consensus algorithms have been proposed for homogeneous and heterogeneous linear networks with jointly connected switching topologies.
These protocols can be explicitly constructed and utilized in a completely distributed manner.
It is shown that the proposed protocols are able to guarantee the achievement of consensus and a strictly positive lower bound for the interval between different triggering instants.
Extending these results to general directed switching graphs or fixed-time consensus \cite{bo-da-ning2017distributed,zong-yu-zuo2018an} is an interesting work in the future.

\ifCLASSOPTIONcaptionsoff
  \newpage
\fi

\bibliography{T_Cyber_bibfile}

\end{document}